\documentclass[12pt]{article}
\usepackage{amsmath,amssymb}
\usepackage{amsthm}
\usepackage{graphicx}
\usepackage{latexsym}
\usepackage{booktabs}
\usepackage{mathrsfs}
\usepackage{amsfonts}
\usepackage{geometry}
\renewcommand{\baselinestretch}{1.5}

\usepackage{arydshln}
\usepackage{natbib}
 \theoremstyle{lemma} \newtheorem{lemma}{Lemma}
\theoremstyle{corollary} \newtheorem{corollary}{Corollary}
\theoremstyle{theorem} \newtheorem{theorem}{Theorem}
\theoremstyle{definition}
\theoremstyle{remark}\newtheorem{remark}{Remark}
 \theoremstyle{proposition}\newtheorem{proposition}{Proposition}
\theoremstyle{example}
   \DeclareMathOperator{\sign}{sign}

\renewcommand{\Pr}{\mathrm{pr}}
\def\T{{ \mathrm{\scriptscriptstyle T} }}

\newcommand{\dma}{\mu_{+A}-\mu_{-A}}

\newcommand{\dhma}{\hat{\mu}_{+A}-\hat{\mu}_{-A}}

\def\T{{ \mathrm{\scriptscriptstyle T} }}

\DeclareMathOperator{\Cov}{Cov}

\title{Sparse Semiparametric Discriminant Analysis}
\author{Qing Mai\\
Department of Statistics, Florida State University\\
Tallahassee, FL, 32306-4330, U.S.A.\\
mai@stat.fsu.edu\\
\\
Hui Zou\\
School of Statistics, University of Minnesota\\
Minneapolis, MN, 55455, U.S.A.\\
zouxx019@umn.edu\\
}

\date{}

\begin{document}
\maketitle
\begin{abstract}
In recent years, a considerable amount of work has been devoted to generalizing
linear discriminant analysis to overcome its incompetence for
high-dimensional classification \citep{WT_2011,CL2011,SLDA,ROAD}. In this paper, we develop high-dimensional
sparse semiparametric discriminant analysis (SSDA) that
generalizes the normal-theory discriminant analysis in two ways: it
relaxes the Gaussian assumptions and can handle ultra-high dimensional
classification problems. If the underlying Bayes rule is sparse,
SSDA can estimate the Bayes rule and select the true features
simultaneously with overwhelming probability, as long as the
logarithm of dimension grows slower than the cube root of sample
size. Simulated and real examples are used to demonstrate the finite
sample performance of SSDA. At the core of the theory is a new
exponential concentration bound for semiparametric Gaussian copulas,
which is of independent interest.
\end{abstract}
Keywords: Gaussian copulas, Linear discriminant analysis,
high-dimension asymptotics, Semiparametric model.

\newpage

\section{Introduction}

Despite its simplicity, linear discriminant analysis (LDA) has
proved to be a valuable classifier in many applications
\citep{STATLOG,Hand_2006}. Let $X=(x_1,\ldots,x_p)$ denote the
predictor vector and $Y\in \{+1,-1\}$ be the class label. The LDA
model states that $X \mid Y \sim N(\mu_Y, \Sigma), $ yielding the
Bayes rule
\begin{equation*}\label{LDA}
\hat
Y^{\mathrm{Bayes}}=\sign\left[\{X-(\mu_++\mu_-)/2\}^\T{\Sigma}^{-1}(\mu_+-\mu_-)+\log{(\pi_+/\pi_-)}\right],
\end{equation*}
where $\pi_y=\Pr(Y=y)$. Given $n$ observations $(Y^{i},X^{i}), 1
\le i \le n$, the classical LDA classifier estimates the Bayes
rule by substituting $\Sigma$, $\mu_y$ and $\pi_y$ with their
sample estimates. As is well known, the classical LDA fails to
cope with high-dimensional data where the dimension, $p$, can be
much larger than the sample size, $n$. A considerable amount of
work has been devoted to generalizing LDA to meet the
high-dimensional challenges. It is generally agreed that
effectively exploiting sparsity is a key to the success of a
generalized LDA classifier for high-dimensional data. Early
attempts include the nearest shrunken centroids classifier (NSC)
\citep{Tibshirani_shrunken} and later the features annealed
independence rule (FAIR) \citep{FF_2008}. These two methods
basically follow the diagonal LDA paradigm with an added variable
selection component, where correlations among variable are completely ignored.
Recently, more sophisticated sparse LDA proposals have been proposed; see \citet{TJ}, \citet{Wu08}, \citet{CLT_2008}, \citet{WT_2011},
\citet{SLDA}, \citet{shao}, \citet{CL2011} and
\citet{ROAD}. In these papers, a lot of empirical and theoretical
results have been provided to demonstrate the competitive
performance of sparse LDA for high-dimensional classification.
These research efforts are rejuvenating discriminant analysis.

However, the existing  sparse LDA methods become ineffective for non-normal data, which is easy to see from the theoretical viewpoint.
See also empirical evidence given in Section 5.1. In the lower dimensional classification problems,
some researchers have considered ways to relax the Gaussian distribution assumption. For example, \citet{mda} proposed the
mixture discriminant analysis (MDA) that uses a mixture of Gaussian
distributions to model the conditional densities of variables given the class label. MDA is
estimated by the Expectation-Maximization algorithm. MDA is a
nonparametric generalization of LDA, but it is not clear how to
further extend MDA to the high-dimensional classification setting with the ability to do variable selection.
\citet{LinJeon} proposed an interesting semiparametric linear discriminant analysis (SeLDA) model. Their model assumes that
after a set of unknown monotone univariate transformations the observed data follow the classical LDA model.
\citet{LinJeon} further showed that the SeLDA model can be accurately
estimated when $p$ is fixed and $n$ goes to infinity. However, the estimator in \citet{LinJeon}
cannot handle high-dimensional classification problems, especially
when $p$ exceeds $n$.

In this paper, we develop high-dimensional sparse semiparametric discriminant analysis (SSDA), a generalization of SeLDA for
high-dimensional classification and variable selection. In particular, we propose a new estimator for the transformation
function and establish its uniform consistency property as long as the logarithm of $p$ is smaller than the cube root of $n$. With the
new transformation estimator, we can transform the data and fit a sparse LDA classifier. In this work we use the direct sparse
discriminant analysis (DSDA) developed by \citet{SLDA}. SSDA enjoys great computational efficiency: its computational complexity
grows linearly with $p$. We show that, if the Bayes rule of the SeLDA model is sparse, then SSDA can consistently select the
important variables and estimate the Bayes rule. At the core of the theory is an exponential concentration bound for semiparametric Gaussian copulas,
which is of independent interest.


\section{Semiparametric LDA Model}

Consider the binary classification problem where we have observed
$n$ random pairs $(Y^{i},X^{i}), 1 \le i \le n$ and wish to classify
$Y$ using a function of $X$. \citet{LinJeon} proposed the following
semiparametric LDA (SeLDA) model that assumes that
\begin{equation}\label{SeLDA}
\left(h_1(X_1),\cdots,h_p(X_p)\right)\mid Y \sim N(\mu_{Y}, \Sigma),
\end{equation}
where $h=(h_1, \cdots, h_p)$ is a set of strictly monotone
univariate transformations. It is important to note that the SeLDA
model does not assume that these univariate transformations are
known or have any parametric forms. By properties of the Gaussian
distribution, $h$ is only unique up to location and scale shifts.
Therefore, for identifiability, assume that $\mu_+=0$,
$\Sigma_{jj}=1, 1 \le j \le p$. The Bayes rule of the SeLDA model is
\begin{equation*}\label{eq:SeLDA}
\hat
Y^{\mathrm{Bayes}}=\sign\left[\{h(X)-(\mu_++\mu_-)/2\}^\T\Sigma^{-1}(\mu_+-\mu_-)+\log{(\pi_+/\pi_-)}\right].
\end{equation*}

The SeLDA model is a very natural generalization of the LDA model.
It is equivalent to modelling the within-group distributions with
semiparametric Gaussian copulas. For any
continuous univariate random variable, $W$, we have
\begin{equation}\label{univariate}
\Phi^{-1}\circ F(W)\sim N(0,1),
\end{equation}
where $F$ is the cumulative probability function (CDF) of $W$ and
$\Phi$ is the CDF of the standard normal distribution. Gaussian copula is a multivariate generalization
of that simple fact of univariate case. Semiparametric Gaussian copula has generated a lot of research interests in recent years;
see \cite{wellner1997}, \cite{song2000}, \cite{tsukahara2005}, \cite{chen2006} and \cite{chen2006jasa}.
The SeLDA model is the first application of semiparametric Gaussian copula in the context of classification.

The following lemma relates the univariate transformation function to the univariate marginal CDF of each predictor.

{\begin{lemma}\label{multi-trans}
Consider a random vector $(X_1,\ldots,X_p)$ with strictly increasing marginal CDFs $F_1,\ldots,F_p$. If there exists a set of strictly increasing univariate functions $h=(h_1,\ldots,h_p)$ such that $h(X)\sim N(0,\Sigma)$, we must have $h_j=\Phi^{-1}\circ F_j$.
\end{lemma}}

{In light of Lemma~\ref{multi-trans},} the SeLDA model can be estimated in
the low-dimensional setting. The basic idea is
straightforward: we first find $\hat h_j(\cdot)$ as good estimates
of these univariate transformation functions and then fit the LDA
model on the ``pseudo data" $\left\{Y^{i},\hat h(X^{i})\right\}, 1 \le
i \le n$. To be more specific, in seek of $\hat h_j$, we let
$F_{+j},F_{-j}$ be the CDF of $X_j$ conditional on $Y=+1$ and
$Y=-1$, respectively, and then we have
\begin{equation*}\label{eq:h}
h_j=\Phi^{-1}\circ F_{+j}=\Phi^{-1}\circ F_{-j}+\mu_-.
\end{equation*}
It can be seen that we only need an estimate of $F_{+j}$ or $F_{-j}$. Denote $n_+,n_-$ as the sample size within the positive and the negative class, respectively.
For convenience, we let $n_{+} \ge n_{-}$ throughout this paper. In other words, we code the class label of the majority class as ``$+1$"
and the minority class as ``$-1$". 

Denote $X_{yj}$ as the $j$th entry of an observation
$X$ belonging to the group $Y=y$, and $\tilde F_{+j}$ as the
empirical CDF of $X_{+j}$. Note that, we cannot directly plug in
$\tilde F_{+j}$ so that $\hat h_j=\Phi^{-1}\circ \tilde F_{+j}$,
because infinite values would occur at tails. Instead, $\tilde
F_{+j}$ is Winsorized at a predefined pair of numbers $(a,b)$ to
obtain $\hat F_{+j}^{a,b}$
\begin{equation}\label{Fhat}
\hat F_{+j}^{a,b}(x)=\left\{\begin{array}{cl} b&\mbox{if $\tilde F_{+j}(x) > b$;}\\
\tilde F_{+j}(x)&\mbox{if $a \le \tilde F_{+j}(x) \le b$;}\\
a&\mbox{if $\tilde F_{+j}(x) <a$.}\end{array} \right.
\end{equation}
Then
\begin{equation}\label{SeLDA:h}
  \hat h_j=\Phi^{-1}\circ \hat F_{+j}^{a,b}.
\end{equation}
 The Winsorization can be viewed as a bias-variance
trade-off.

With $\hat h_j$, the covariance matrix $\Sigma$ is estimated by the pooled
sample covariance matrix of $\hat h(X^i)$ and $\mu_{-j}$ is estimated
by
\begin{eqnarray*}\label{SeLDA:mu}
\hat \mu_{-j}&=&q^{-1}[n_-^{-1}\sum_{i=1}^{n_-}\hat
h(X_{-j}^i)\mathrm{1}_{\tilde F(X_{-j}^i)\in
(a,b)}+\phi\{\Phi^{-1}\circ \tilde F_{-j}\circ \tilde
F_{+j}^{-1}(b)\}-\phi\{\Phi^{-1}\circ \tilde F_{-j}\circ
\tilde F_{+j}^{-1}(a)\}]
\end{eqnarray*}
where $\phi$ is the density function for a standard normal random
variable and
\begin{equation*}
q=n_-^{-1}\sum_{i=1}^{n_-}\mathrm{1}_{\tilde
F_{+j}(X^i)\in (a,b)}.
\end{equation*}
$\hat \mu_{-j}$ has this complicated form because of the
Winsorization. \cite{LinJeon} showed that when $p$ is fixed and $n$
tends to infinity, $\hat \Sigma$, $\hat \mu_-$ are consistent.

\section{Estimation of The High-dimensional Semiparametric LDA Model}

We need to address two technical problems when applying the SeLDA model to
high-dimensional classification. First, we must modify the estimator in
\citet{LinJeon} to achieve consistency under ultra-high dimensions.
Second, the SeLDA model is not estimable with the large-$p$-small-$n$ data,
even when we know the true transformation functions.
To overcome this difficulty, we propose to fit a sparse SeLDA model by exploiting a sparsity assumption on the underlying
Bayes rule. For the sake of presentation, we first discuss how to fit a sparse SeLDA model, provided that good estimators of
$h_j(\cdot), 1 \le j \le p$, are already obtained.
After
introducing the sparse SeLDA, we focus on a new strategy to estimate
$h_j(\cdot), 1 \le j \le p $.

\subsection{Exploiting sparsity}

We assume that the Bayes rule of the SeLDA model only involves a
small number of predictors. To be more specific, let $\beta^{\mathrm
{Bayes}}=\Sigma^{-1}(\mu_+-\mu_-)$ and define $A=\{j: \beta^{\mathrm
{Bayes}}_j \neq 0\}$. Sparsity means that $\vert A \vert \ll p.$ An
elegant feature of SeLDA is that it keeps the interpretation of LDA,
that is, variable $j$ is irrelevant if and only if
$\beta_j^{\mathrm{Bayes}} =0$.

Suppose that we have obtained $\hat h_j(\cdot)$ as a good estimate
of $h_j(\cdot), 1 \le j \le p$, we focus on estimating the sparse
LDA model using the ``pseudo data" $\{Y^{i},\hat h(X^{i})\}, 1 \le i
\le n$. Among the previously mentioned sparse LDA proposals in the literature, only
\cite{FF_2008}, \cite{shao}, \cite{CL2011}, \cite{SLDA} and \cite{ROAD} provided theoretical analysis
of their methods. \cite{FF_2008}'s theory assumes that $\Sigma$ is a diagonal matrix.
\cite{shao}'s method works well only under some strong sparsity assumptions on the covariance matrix $\Sigma$ and $\mu_+-\mu_-$.
The sparse LDA methods proposed in \cite{CL2011}, \cite{SLDA} and \cite{ROAD} are shown to work well under general
correlation structures. From the computational perspective, the method in
\cite{SLDA} is most computationally efficient. Therefore, it is the method used here to exploit sparsity.

The proposal in \citet{SLDA}, which is referred to as DSDA, begins with the observation that the classical LDA direction can be exactly recovered by doing linear
regression of $Y$ on $h(X)$ \citep{HTF_2008} where $Y$ is treated as
a numeric variable. Define $\Omega=\Cov\left\{h(X)\right\}$, and
$\beta^{\star}=\Omega^{-1}(\mu_+-\mu_-),\beta^{\mathrm{Bayes}}=\Sigma^{-1}(\mu_+-\mu_-)$. It can be shown that
$\beta^{\star}$ and $\beta^{\mathrm{Bayes}}$ have the same direction.
For variable selection and classification, it suffices to estimate $\beta^{\star}$.
DSDA aims at estimating $\beta^{\star}$ by the following penalized least squares approach:
\begin{eqnarray}\label{SLDA}
\hat
\beta^{\mathrm{DSDA}}&=&\arg\min_{\beta}[n^{-1}\sum^n_{i=1}\left\{Y^i-\beta_0-
h(X^{i})^\T\beta\right\}^2+\sum^p_{j=1}P_{\lambda}
(|\beta_j|)],\\
\hat \beta^{\mathrm{DSDA}}_0&=&-(\hat \mu_++\hat
\mu_-)^\T\hat\beta^{\mathrm{DSDA}}/2+\log{(\hat
\pi_+/\hat\pi_-)}\cdot{({\hat{\beta}^{\mathrm{DSDA}}})^\T}\hat\Sigma\hat\beta^{\mathrm{DSDA}}/\{(\hat\mu_+-\hat\mu_-)^\T\hat\beta^{\mathrm{DSDA}}\} \nonumber
\end{eqnarray}
where, under the LDA model, $h$ is known to be $h(X)=X$, and
$P_{\lambda}(\cdot)$ is a sparsity-inducing penalty, such as Lasso
\citep{Tibs96a} or SCAD \citep{FL_2001}. Then the DSDA classifier is
$\sign\left\{\hat \beta_0^{\mathrm{DSDA}}+h(X)^\T\hat
\beta^{\mathrm{DSDA}}\right\}$.
There are many other penalty functions proposed for sparse
regression, including the elastic net \citep{ZH_2005}, the adaptive
lasso \citep{Zou_2006}, SICA \citep{LvFan} and the MCP
\citep{Zhang}, among others. All these penalties can be used in DSDA. The original paper \citep{SLDA} used the Lasso penalty where $P_\lambda(t)=\lambda t$ for
$t>0$. One could use either $\tt lars$ \citep{EHIT_2004} or $\tt
glmnet$ \citep{FHT_2008} to efficiently implement DSDA.

If we knew these transformation functions $h$ in the SeLDA model,
(\ref{SLDA}) could be directly used to estimate the Bayes rule of
SeLDA. In SSDA we substitute $h_j$ with its estimator $\hat h_j$
and apply sparse LDA methods to $(Y,\hat h(X))$. For example, to use DSDA in the SeLDA model, we solve for
\begin{eqnarray}\label{SSeLDA}
\hat \beta&=&\arg\min_{\beta}[n^{-1}\sum^n_{i=1}\left\{Y^i-\beta_0-
\hat h(X^{i})^\T\beta\right\}^2+\sum^p_{j=1}P_{\lambda}
(|\beta_j|)],\\
\hat \beta_0&=&-(\hat \mu_++\hat
\mu_-)^\T\hat\beta/2+\log{(\hat
\pi_+/\hat\pi_-)}\cdot{{\hat\beta}^\T}\hat\Sigma\hat\beta/\{(\hat\mu_+-\hat\mu_-)^\T\hat\beta\} \nonumber
\end{eqnarray}
Then (\ref{SSeLDA}) yields the SSDA classification rule: $\sign\left\{\hat \beta_0+\hat
h(X)^\T\hat\beta\right\}$.

\subsection{Uniform estimation of transformation functions}

We propose a high-quality estimator of the monotone
transformation function.
In order to establish the theoretical property of SSDA, we need
all $p$ estimators of the transformation function to uniformly
converge to the truth at a certain fast rate, even when $p$ is much
larger than $n$. Our estimator is defined as
\begin{equation}\label{Fhat1}
\hat F_{+j}(x)=\left\{\begin{array}{cl} 1-1/n_+^2&\mbox{if $\tilde F_{+j}(x) > 1-1/n_+^2$}\\
\tilde F_{+j}(x)&\mbox{if $1/n_+^2 \le \tilde F_{+j}(x) \le 1-1/n_+^2$}\\
1/n_+^2&\mbox{if $\tilde F_{+j}(x)
<1/n_+^2$}\end{array} \right.
\end{equation}
and then
\begin{equation*}
\hat h_j=\Phi^{-1}\circ \hat F_{+j}.
\end{equation*}
Note that the class with a bigger size is coded as ``$+$" as mentioned in section 2.

In other words, instead of fixing the Winsorization parameters $a,b$ as in \eqref{Fhat}, we let
\begin{equation}\label{ab}
(a,b)=(a_n,b_n)=(1/n_+^2, 1-1/n_+^2).
\end{equation}
 With the presence of $\Phi^{-1}$, it is necessary to choose $a_n>0,b_n<1$ to avoid extreme values at tails. On the other hand, $a_n\rightarrow 0, b_n\rightarrow 1$ so that the bias will automatically vanish as $n\rightarrow \infty$. To further see that \eqref{ab} are proper choices of $a_n,b_n$, see the theory developed in Section 3 for mathematical justification.

Other estimators have been proposed. For example, \citet{SemiCov} considered a one-class problem
with Gaussian copulas, which essentially states $h(X)\sim
N(0,\Sigma)$, and aims to estimate $\Sigma^{-1}$. In their paper,
$h_j$ is estimated by $\hat h_j=\Phi^{-1}\circ \hat F^{a_n,b_n}$,
where
$a_n=1-b_n=(4n^{1/4}\sqrt{\pi\log{n}})^{-1}.$
\citet{SemiCov} showed that this estimator is consistent when $p$ is smaller than any polynomial order of $n$, but  it is not clear
whether the final SSDA can handle non-polynomial high
dimensions.

\begin{remark} {\it
Rank-based estimators were independently proposed by \citet{semigraph,XZ} for estimating $\Sigma^{-1}$ without estimating the transformation functions.
However, in the discriminant analysis problem considered here we need to estimate both $\Sigma^{-1}$ and the mean vectors. The estimation of the mean vectors requires to estimate
the transformation functions.}
\end{remark}

\subsection{The pooled transformation estimator}

We now consider an estimator that pools information from both classes. According to \eqref{eq:h}, we can find the estimated transformation functions by choosing proper $\hat F_{+j}$ and/or $\hat F_{-j},\hat\mu_{-j}$. The naive estimate only uses the data from the positive class because of the difficulty in estimating $\mu_{-j}$. However, we have the following lemma that will assist us in developing a more sophisticated transformation estimation utilizing all the data points.

\begin{lemma}\label{h:twoclass}
Consider the model in \eqref{SeLDA}. Then we have
\begin{enumerate}
\item Conditional on $Y=-1$, we have
\begin{equation*}
E(\Phi^{-1}\circ F_{+j}(X_j))=\mu_{-j}.
\end{equation*}
\item Conditional on $Y=+1$, we have
\begin{equation*}
E(\Phi^{-1}\circ F_{-j}(X_j))=-\mu_{-j}.
\end{equation*}
\end{enumerate}
\end{lemma}

Set $\hat F_{+j}$ as defined in \eqref{Fhat1} and $\hat F_{-j}$ as the empirical CDF for $X_j$ conditional on $Y=-1$ Winsorized at $(a_{-n},b_{-n})=(1/n_-^2,1-1/n_-^2)$. Then by Lemma~\ref{h:twoclass}, we can define a pooled estimator of $\mu_{-j}$:
\begin{equation*}
\hat\mu_{-j}^{\mathrm{(pool)}}=\hat\pi_+\hat \mu_{-j}^{(+)}+\hat\pi_-\hat\mu_{-j}^{(-)},
\end{equation*}  
where 
\begin{eqnarray*}
\hat \mu_{-j}^{(+)}&=&\frac{1}{n_-}\sum_{Y^i=-1}\Phi^{-1}\circ \hat F_{+j}(X_j^i)\\
\hat \mu_{-j}^{(-)}&=&-\frac{1}{n_+}\sum_{Y^i=+1}\Phi^{-1}\circ \hat F_{-j}(X_j^i)
\end{eqnarray*}
Then
\begin{equation*}
\hat h_j^{\mathrm{(pool)}}=\hat\pi_+ \hat h_j^{(+)}+\hat\pi_- \hat h_j^{(-)},
\end{equation*}
where 
\begin{eqnarray*}
\hat h_j^{(+)}&=&\Phi^{-1}\circ \hat F_{+j}\\
\hat h_j^{(-)}&=&\Phi^{-1}\circ \hat F_{-j}+\hat\mu_{-j}^{\mathrm{(pool)}}
\end{eqnarray*}
This estimator utilizes all the data points. We refer to this estimator as the pooled estimator. In Section 5 we will present numerical evidence that the pooled estimator does improve over the naive estimator in many cases.


\section{Theoretical Results}

\subsection{Estimation of transformation functions}
To explore the consistency property of SSDA, we first study the
estimation accuracy of semiparametric Gaussian copulas. The results in this
subsection are applicable to any statistical model using semiparametric Gaussian copulas,
which is of independent interest itself. Consider the one-class estimation case first. Assume
that $X$ is a $p$-dimensional random variable such that $h(X)\sim
N(\mathrm{0}_p, \Sigma)$ with $h_j=\Phi^{-1}\circ F_j$ and $\hat
h_j=\Phi^{-1}\circ \hat F_{j}$, where $\hat F_{j}$ is defined as in
(\ref{Fhat1}). Denote $\hat \mu_j$ and $\hat \sigma_{jk}$ as the
sample mean and sample covariance for corresponding features. We
establish exponential concentration bounds for $\hat \mu_j$ and
$\hat \sigma_{jk}$. For writing convenience, we use $c$ to denote generic constants throughtout.

\begin{theorem}\label{accuracy}
Define
\begin{eqnarray*}
\zeta_1^*(\epsilon)&=&
2\exp(-cn\epsilon^2)+4\exp(-cn^{1-2\rho}\epsilon^2/\rho)+4\exp(-cn^{\frac{1}{2}-\rho}),\\
\zeta_2^*(\epsilon)&=&
c\exp(-cn\epsilon^2)+c\exp(-cn^{\frac{1}{3}-\rho})+c\exp(-cn^{1-\rho})+c\exp\{-c(\log^2n)n^{1-2\rho}\epsilon^2/\rho^2\}.
\end{eqnarray*}
 For sufficiently large $n$
and any $0<\rho<\frac{1}{3}$, there exists a positive constant
$\epsilon_0$ such that, for any $0<\epsilon<\epsilon_0$, we have
\begin{eqnarray}
\Pr(|\hat \mu_j-\mu_j|>\epsilon)&\le& \zeta_1^*(\epsilon)\label{eq:muhat}\\
\Pr(|\hat \sigma_{jk}-\sigma_{jk}|>\epsilon)&\le&
\zeta_2^*(\epsilon)\label{eq:sigmahat}
\end{eqnarray}
\end{theorem}

For the two-class SeLDA model, we can easily obtain the following
corollary from Theorem~\ref{accuracy}.

\begin{corollary}\label{twoclass}
Define
\begin{eqnarray}
\zeta_1(\epsilon)&=&\zeta_1^*(\pi_+^{1/2}\epsilon/2)+\zeta_1^*(\pi_-^{1/2}\epsilon/2)+4\exp(-cn)\label{SSeLDA:mu}\\
\zeta_2(\epsilon)&=&\zeta_2^*(\pi_+^{1/2}\epsilon/2)+\zeta_2^*(\pi_-^{1/2}\epsilon/2)+4\exp(-cn)+2\zeta_1(\epsilon)\label{SeLDA:sigma}
\end{eqnarray}
Then there exists a positive constant $\epsilon_0$ such that, for any $0<\epsilon<\epsilon_0$, we have
\begin{eqnarray*}
&&\Pr(|(\hat \mu_{+j}-\hat \mu_{-j})-(\mu_{+j}-\mu_{-j})|>\epsilon)\le \zeta_1(\epsilon)\label{eq:dmuhat}\\
&&\Pr(|\hat \sigma_{jk}-\sigma_{jk}|>\epsilon)\le
\zeta_2(\epsilon)\label{eq:Chat}
\end{eqnarray*}
\end{corollary}

\begin{remark} {\it
Theorem~\ref{accuracy} and Corollary 1 can
be used for other high-dimensional statistical problems involving
semiparametric Gaussian copulas.}
\end{remark}


\subsection{Consistency of SSDA} 
In this section, we study the theoretical results for SSDA. For simplicity, we focus on SSDA with the naive estimator of transformation functions. The theoretical properties for SSDA combined with the pooled transformation estimator can be derived similarly with more lengthy calculation of the probability bounds.

With the results in Section 4.1, we are ready to prove the rate of convergence of
SSDA. We first define necessary notation. Define
$\beta^{\star}=\Omega^{-1}(\mu_+-\mu_-)$, where $\Omega$ is the covariance of $X$. Recall that $\beta^{\star}$ is
equal to $c\Sigma^{-1}(\mu_+-\mu_-)=c\beta^{\mathrm{Bayes}}$ for
some positive constant \citep{SLDA}. Then we can write $A=\{j:
\beta_j^{\star} \neq 0\}$. Let $s$ be the cardinality of $A$. In
addition, for an $m_1\times m_2$ matrix $M$, denote $\Vert M
\Vert_{\infty}=\max_{i=1,..,m_1} \sum^{m_2}_{j=1} |M_{ij}|$, and,
for a vector $u$, $\Vert u \Vert_{\infty}=\max |u_j|$. Throughout the proof, we assume that $s\ll n^{1/4}$. Define the
following quantities that are repeatedly used:
\begin{eqnarray*} &&\kappa=\Vert
\Omega_{A^cA}(\Omega_{AA})^{-1}\Vert_{\infty}, \quad \varphi=\Vert
(\Omega_{AA})^{-1} \Vert_{\infty},\Delta=\Vert \dma\Vert_{\infty},\\
&& \Delta_1=\Vert \dma\Vert_1, \quad \Delta_2=\Vert\mu_{+A}+\mu_{-A}\Vert_{\infty}
\quad \nu =\min_{j \in A}|\beta_j|/\Delta \varphi .
\end{eqnarray*}

 Suppose that the lasso estimator correctly shrinks
$\hat \beta_{A^c}$ to zero, then SSDA should be equivalent to
performing SeLDA on $X_A$. Therefore, define the hypothetical
estimator
\begin{equation*}\label{hyplasso}
\hat
\beta_{A}^{\mathrm{hyp}}=\arg\min_{\beta,\beta_0}[n^{-1}\sum^n_{i=1}\{Y^i-\beta_0-\sum_{j
\in A}\hat h_j(X^i_j)\beta_j\}^2+\sum_{j \in A} \lambda
|\beta_j|].
\end{equation*}
Then, we wish that $\hat \beta=(\hat
\beta_A^{\mathrm{hyp}},\mathrm{0}_{A^c})$ with $\hat \beta^{\mathrm{hyp}}_j\ne 0$
for $j \in A$. To ensure the consistency of SSDA, we further
require the following condition:
\begin{equation}\label{lasso-condition}
\quad \kappa=\Vert \Omega_{A^cA}(\Omega_{AA})^{-1} \Vert_{\infty} <1.
\end{equation}
The condition in (\ref{lasso-condition}) is an analogue of the
ir-representable condition for the lasso penalized linear
regression model
\citep{Nicolai06,Zou_2006,Yu,Martin09}. Weaker conditions exist if one is only concerned with oracle inequalities for the coefficients under the regression model, such as the restricted eigenvalue condition \citep{BRT09,vB09,RWY10}. It would be interesting to investigate if similar conditions can be extended to the framework of SSDA. We leave this as a future project. On the other hand, if one is reluctant to assume \eqref{lasso-condition}, the use of a concave penalty, such as SCAD \citep{FL_2001}, can remove this condition; see the discussion in \cite{SLDA}.

\begin{theorem}\label{thm1}
Define $\zeta_1,\zeta_2$ as in Corollary~\ref{twoclass}. Pick any $\lambda$ such that $\lambda <
{\min\{\min_{j\in A}{|\beta_j|}/(2\varphi),\Delta\}}.$ Then for any $\epsilon>0$ and sufficiently large $n$ such that $\epsilon> csn^{-\rho/2}$, where $c$ does not depend on $(n,p,s)$, we have

1. Assuming the condition in (\ref{lasso-condition}), with
probability at least $1-\psi_1$, $\hat \beta_A = \hat
\beta_A^{\mathrm{hyp}}$ and $\hat \beta_{A^c}=0$, where
\begin{equation*}\label{thm1.1}
\psi_1=2ps\zeta_2(\epsilon/s)+2p\zeta_1[
\lambda(1-\kappa-2\epsilon \varphi)/\{4(1+\kappa)\}]
\end{equation*}
and $\epsilon$ is any positive constant less than {$\min\left(
\epsilon_0,
\lambda(1-\kappa)/[4\varphi\left\{\lambda/2+(1+\kappa)\right\}\Delta]\right)$}.

2. With probability at least $1-\psi_2$, none of the elements of
$\hat \beta_A$ is zero, where
\begin{equation*}\label{thm2.1}
\psi_2=2s^2\zeta_2(\epsilon/s)+2s\zeta_1(\epsilon)
\end{equation*}
and $\epsilon$ is any positive constant less than
{$\min\left\{\epsilon_0,\nu/\{ (3+\nu)\varphi\}, \Delta
\nu/( 6+2\nu) \right\}$.}

3. For any positive $\epsilon$ satisfying
$\epsilon<\min\left\{\epsilon_0,\lambda/( 2\varphi\Delta),
\lambda\right\}$, we have
\begin{eqnarray*}\label{thm2.2} \Pr(\Vert \hat
\beta_A-\beta_A\Vert_{\infty} \le 4\varphi\lambda) & \ge &
1-2s^2\zeta_2(\epsilon/s)-2s\zeta_1(\epsilon).
\end{eqnarray*}

\end{theorem}

Theorem \ref{thm1} provides the foundation for asymptotic results.
Assume the following two regularity conditions.

(C1). $n,p \rightarrow \infty$ and
$s^2\log(ps)/n^{1/3-\rho}\rightarrow 0$, for some $\rho$ in
$(0, 1/3)$;

(C2).  $\min_{j\in A}|\beta_{j}|
\gg\max[sn^{-\rho/2},s\{\log(ps)/n^{1/3-\rho}\}^{1/2}]$ for some $\rho$ in
$(0, 1/3).$

Condition (C1) restricts that $p,s$ should not grow too fast
comparing to $n$. However, $p$ is allowed to grow faster than any polynomial order of $n$.
Condition (C2) states that the important features
should be sufficiently large such that we can separate them from the
noises, which is a standard assumption in the literature of sparse recovery. The next theorem shows that
SSDA consistently recovers the Bayes rule of the SeLDA model.

\begin{theorem}\label{asy}
 Let
$\widehat A=\{j: \hat \beta_j \neq 0\}$. Under conditions (C1) and
(C2), if we choose $\lambda=\lambda_n$ such that $\lambda_n \ll
\min_{j\in A}|\beta_{j}|$ and $\lambda_n \gg
s\{\log(ps)/n^{\frac{1}{3}-\rho}\}^{1/2}$, and further
assume $\kappa <1$, then $\Pr(\widehat A =A) \rightarrow 1$ and
$\Pr\left(\Vert \hat \beta_A-\beta_A\Vert_{\infty} \le 4\varphi
\lambda_n \right) \rightarrow 1$.

\end{theorem}

Further, we prove that SSDA is asymptotically equivalent to the Bayes rule in terms of error rate . Note that the Bayes error rate $R=\Pr(Y\ne \sign(h(X)^\T\beta^*+\beta_0))$ and $R_n=\Pr(Y\ne \sign(\hat h(X)^\T\hat\beta+\hat\beta_0))$. We have the following theorem.

\begin{theorem}\label{error}
Define $\zeta_1,\zeta_2$ as in Corollary~\ref{twoclass}. Pick any $\lambda$ such that $\lambda <
{\min\{\dfrac{\min_{j\in A}{|\beta_j|}}{2\varphi},\Delta\}}.$ Then for a sufficiently small constant $\epsilon>0$ and sufficiently large $n$ such that $\epsilon>csn^{-\rho/2}$, where $c$ does not depend on $(n,p,s)$, with probability no smaller than $1-\psi_3$, we have $R_n-R<\epsilon$, where
\begin{equation}
\psi_3= cs\zeta_1(\dfrac{\epsilon}{s(\phi\Delta_1+\Delta_2)})+cp\zeta_1(\dfrac{\lambda(1-\kappa+2\epsilon\phi)}{4(1+\kappa)})+2ps\zeta_2(\dfrac{c\epsilon}{s})+cp\exp(-c\dfrac{n^{1-\rho}}{\rho\log{n}}).
\end{equation}
\end{theorem}

\begin{corollary}
Under conditions (C1) and
(C2), if we choose $\lambda=\lambda_n$ such that $\lambda_n \ll
\min_{j\in A}|\beta_{j}|$ and $\lambda_n \gg
\sqrt{\log(ps)\dfrac{s^2}{n^{\frac{1}{3}-\rho}}}$, and further
assume $\kappa <1$, then
\begin{equation}
R_n-R\rightarrow 0 \quad\text{ in probability}
\end{equation}
\end{corollary}
\begin{remark}{\it
Our results concerning the error rate of SSDA are much more involved than those for sparse LDA algorithms in \citet{CL2011, ROAD}, because of the semiparametric assumptions. Under the parametric LDA model, the error rate tends to the Bayes error as long as the discriminant direction $\beta$ is estimated consistently. However, under the SeLDA model, we deal with the extra uncertainty in estimating $h$ and need some uniform convergence results on $\hat h(X)$.}
\end{remark}

\section{Numerical Results}
\subsection{Simulation}

We examine the finite sample performance of SSDA by simulation. We consider two transformation estimators: the naive estimator and the pooled estimator. The resulting methods are denoted by SSDA(naive) and SSDA(pooled), respectively.
For comparison, in the simulation study  we also include DSDA and
the sparse LDA algorithm \citep{WT_2011} denoted by Witten for
presentation purpose. After we apply the estimated transformation to
the data, we use Witten's sparse LDA algorithm to fit the
classifier. This gives us Se-Witten, another competitor in the
simulation study.


Four types of SeLDA models were considered in the study. In each
model, we first generated $Y$ with $\pi_+=\pi_-=0.5$. For convenience, we say that $\Sigma$ has AR($\rho$) structure if $\Sigma_{ij}=\rho^{|i-j|}$ and $\Sigma$ has CS($\rho$) structure if $\Sigma_{ij}=\rho$ for any $i\ne j$.
We fixed $\mu_{-}=0$ and $\mu_{+}=\Sigma \beta^{\mathrm{Bayes}}$.

\begin{enumerate}
\item[] Model 1: $n=150$, $p=400$. $\Sigma$ has AR(0.5)
structure.
\begin{equation*}
\beta^{\mathrm{Bayes}}=0.556(3, 1.5, 0,0,2,\mathrm{0}_{p-5})^\T.
\end{equation*}

\item[] Model 2: $n=200$, $p=400$. $\Sigma$ has AR(0.5)
structure.
\begin{equation*}
\beta^{\mathrm{Bayes}}=0.582(3, 2.5, -2.8,\mathrm{0}_{p-3})^\T.
\end{equation*}
\item[] Model 3: $n=400$, $p=800$. $\Sigma$ has CS(0.5)
structure.
\begin{equation*}
\beta^{\mathrm{Bayes}}=0.395(3, 1.7, -2.2,
-2.1,2.55,\mathrm{0}_{p-5})^\T.
\end{equation*}

\item[] Model 4: $n=300$, $p=800$. $\Sigma$ is block diagonal with 5
blocks of dimension $160 \times 160$. Each block has CS(0.6)
structure.
\begin{equation*}
\beta^{\mathrm{Bayes}}=0.916(1.2,-1.4,1.15,-1.64,1.5,-1,2,\mathrm{0}_{p-7})^\T.
\end{equation*}


\end{enumerate}

\renewcommand{\baselinestretch}{1}
\begin{table}
\caption{Choices of $g_j$ in Models 1b--4b.} \centering
\fbox{%
\begin{tabular}{c|c|c|c}
$g_j(v)$&Models 1b,2b&Model 3b&Model 4b\\
&$j$&$j$&$j$\\
 \hline
$v^3$&$1, 101,\ldots, 150$&$1,201,\ldots,300$&$3,201,\ldots,300$\\
$\exp{(v)}$&$2,151, \ldots,
200$&$2,301,\ldots,400$&$4,301,\ldots,400$\\
$\mathrm{arctan}{(v)}$&$3,201,\ldots,300$&$3,401,\ldots,500$&$5,401,\ldots,500$\\
$v^3$&$4,\ldots, 50$&$4,6,\ldots,100$&$1,8,\ldots,100$\\
$\Phi(v)$&$51, \ldots,
100$&$5,101,\ldots,200$&$2,101,\ldots,200$\\
$(v+1)^3$&$301,\ldots, 350$&$501,\ldots,600$&$6,501,\ldots,600$\\
$\mathrm{arctan}(2v)$&$351,\ldots,400$&$601,\ldots,800$&$7,601,\ldots,800$\\
\end{tabular}}
\end{table}
\renewcommand{\baselinestretch}{1.5}
We transform $V$ to $X$ by $X=g(V)$ and the final data to be used are $(X,Y)$.
In each type of model, we consider two sets of $g$. We call the
resulting models series a and b. In series a, $X=V$ so that the
SeLDA model becomes the LDA model. In series b, we considered some commonly used transformations
such that that some features become heavily skewed, some heavy-tailed and some
bounded. The choices of $g$ are listed in Table 1. In the simulation
study we also considered the oracle sparse discriminant classifiers
including oracle DSDA and oracle Witten. The idea is to apply the
true transformation to variables and then fit a sparse LDA
classifier using DSDA or Witten and Tibshirani's method.

The simulation results for Models 1a--4a and Models 1b--4b are reported in Table 2 and Table 3, respectively.
Note that in Table 2 DSDA and Witten are the oracle DSDA and the oracle Witten.
We can draw the following conclusions from Tables 2 and 3.

\begin{itemize}
\item  Models 1a--4a are actually LDA models.
SSDA performs very similarly to DSDA. Although SSDA has
slightly higher error rates, this is expected because SSDA does
not use the parametric assumption. On the other hand, in Models
1b--4b, SSDA performs much better than DSDA. These results
jointly show that SSDA is a much more robust sparse discriminant
analysis algorithm than those based on the LDA model.

\item In both tables, SSDA is very close to the oracle DSDA, which empirically shows the high quality of the proposed transformation estimator in Section 3.2.
In all eight cases, SSDA is a good approximation to the Bayes rule, which is consistent with the theoretical results. On the other hand, in Models 1, 2, 4 (a) \& (b) SSDA(pooled) yields slightly lower error rates, which illustrates the advantage of utilizing the information from both classes when estimating the transformation.

\item Se-Witten is a different SSDA classifier in which Witten
and Tibshirani's method is used to fit the SeLDA model after
estimating the transformation functions. Se-Witten performs very
well in Models 1a,2a,1b,2b but it performs very poorly in Models
3a,4a,3b,4b. The same is true for the oracle Witten method. By
comparing SSDA and Se-Witten, we see that DSDA works better
than Witten and Tibshirani's method. In addition to the theory in
Section 4, the simulation also supports the use of DSDA in fitting
the high-dimensional sparse semiparametric LDA model.
\end{itemize}

\renewcommand{\baselinestretch}{1}
\begin{table}
\caption{Simulation results for Models 1a--4a. The reported numbers
are medians  based on 2000 replications. Their standard errors
obtained by bootstrap are in parentheses. TRUE selection and FALSE
selection denote the numbers of selected important variables and
unimportant variables, respectively. } \centering
\fbox{%
{\small
\begin{tabular}{c|ccccccccccc}
&Bayes&Oracle&SSDA&SSDA&DSDA&Oracle&\multicolumn{2}{c}{Se-Witten}&Witten\\
 &&DSDA&(naive)&(pooled)&&Witten&(naive)&(pooled)\\
\hline Model 1 (a)&\\
Error(\%)&10&10.71&11.5&11.11&10.71&11.39&11.56&11.57&11.39\\
&&(0.02)&(0.03)&(0.03)&(0.02)&(0.02)&(0.01)&(0.02)&(0.02)\\
TRUE selection&3&3&3&3&3&3&3&3&3\\
&&(0)&(0)&(0)&(0)&(0)&(0)&(0)&(0)\\
FALSE selection&0&1&2&2&1&26&26&25&26\\
&&(0.14)&(0.38)&(0.1)&(0.14)&(0.42)&(0.09)&(0)&(0.42)\\
\hline Model 2 (a)&\\
Error(\%)&10&11.09&11.66&11.57&11.09&13.36&13.46&13.58&13.36\\
&&(0.02)&(0.03)&(0.03)&(0.02)&(0.03)&(0.04)&(0.02)&(0.03)\\
TRUE selection&3&3&3&3&3&3&3&3&3\\
&&(0)&(0)&(0)&(0)&(0)&(0)&(0)&(0)\\
FALSE selection&0&5&6&6&5&24&24&24&24\\
&&(0.37)&(0.51)&(0.48)&(0.37)&(0)&(0)&(0.5)&(0)\\
 \hline
Model 3 (a)&\\
Error(\%)&20&21.93&22.13&22.3&21.93&33.69&34.18&35.05&33.69\\
&&(0.03)&(0.03)&(0.03)&(0.03)&(0.01)&(0)&(0)&(0.01)\\
TRUE selection&5&5&5&5&5&3&5&5&3\\
&&(0)&(0)&(0)&(0)&(0)&(0)&(0)&(0)\\
FALSE selection&0&14&13&14&14&419.5&795&795&419.5\\
&&(0.59)&(0.57)&(0.58)&(0.59)&(10.19)&(0)&(0)&(10.19)\\
\hline
Model 4 (a)&\\
Error(\%)&10&12.50&13.20&12.78&12.50&23.90&26.14&26&23.90\\
&&(0.02)&(0.05)&(0.03)&(0.02)&(0.01)&(0.01)&(0.01)&(0.01)\\
TRUE selection&7&7&7&7&7&4&5&5&4\\
&&(0)&(0)&(0)&(0)&(0)&(0.02)&(0)&(0)\\
FALSE selection&0&18&17&17&18&35&153&153&35\\
&&(0.70)&(0.54)&(0.45)&(0.70)&(4.43)&(0)&(0)&(4.43)\\
\end{tabular}}}
\end{table}

\begin{table}
\caption{Simulation results for Models 1b--4b. The reported numbers
are medians  based on 2000 replications. Their standard errors
obtained by bootstrap are in parentheses. TRUE selection and FALSE
selection denote the numbers of selected important variables and
unimportant variables, respectively. } \centering
\fbox{%
{\small
\begin{tabular}{c|ccccccccccc}
&Bayes&Oracle&SSDA&SSDA&DSDA&Oracle&\multicolumn{2}{c}{Se-Witten}&Witten\\
 &&DSDA&(naive)&(pooled)&&Witten&(naive)&(pooled)\\
\hline Model 1 (b)&\\
Error(\%)&10&10.71&11.5&11.11&18.24&11.39&11.56&11.57&16.19\\
&&(0.02)&(0.03)&(0.03)&(0.02)&(0.02)&(0.01)&(0.02)&(0.02)\\
TRUE selection&3&3&3&3&3&3&3&3&3\\
&&(0)&(0)&(0)&(0)&(0)&(0)&(0)&(0)\\
FALSE selection&0&1&2&2&1&26&26&25&25\\
&&(0.14)&(0.38)&(0.1)&(0)&(0.42)&(0.09)&(0)&(0.5)\\
\hline Model 2 (b)&\\
Error(\%)&10&11.09&11.66&11.57&19.47&13.36&13.46&13.58&20.16\\
&&(0.02)&(0.03)&(0.03)&(0.09)&(0.03)&(0.04)&(0.02)&(0.04)\\
TRUE selection&3&3&3&3&3&3&3&3&2\\
&&(0)&(0)&(0)&(0)&(0)&(0)&(0)&(0)\\
FALSE selection&0&5&6&6&5&24&24&24&20\\
&&(0.37)&(0.51)&(0.48)&(0.37)&(0)&(0)&(0.5)&(0.17)\\
 \hline
Model 3 (b)&\\
Error(\%)&20&21.93&22.13&22.3&26.76&33.69&34.18&35.05&34.25\\
&&(0.03)&(0.03)&(0.03)&(0.03)&(0.01)&(0)&(0)&(0)\\
TRUE selection&5&5&5&5&5&3&5&5&3\\
&&(0)&(0)&(0)&(0)&(0)&(0)&(0)&(0)\\
FALSE selection&0&14&13&14&15&419.5&795&795&795\\
&&(0.59)&(0.57)&(0.58)&(0.67)&(10.19)&(0)&(0)&(10.19)\\
\hline
Model 4 (b)&\\
Error(\%)&10&12.50&13.20&12.78&19.88&23.90&26.14&26&26.83\\
&&(0.02)&(0.05)&(0.03)&(0.04)&(0.01)&(0.01)&(0.01)&(0.01)\\
TRUE selection&7&7&7&7&6&4&5&5&6\\
&&(0)&(0)&(0)&(0)&(0)&(0.02)&(0)&(0.23)\\
FALSE selection&0&18&17&17&25&35&153&153&153\\
&&(0.70)&(0.54)&(0.45)&(0.83)&(4.43)&(0)&(0)&(0.09)\\
\end{tabular}}}
\end{table}
\renewcommand{\baselinestretch}{1.5}

\subsection{Malaria data}

We further demonstrate SSDA by using the malaria data \citep{blood}. This dataset is available at
\begin{verbatim}
http://www.ncbi.nlm.nih.gov/sites/GDSbrowser?acc=GDS2362.
\end{verbatim}
Out of 71 samples in the dataset, 49 have been infected with
malaria, while 22 are healthy people. The predictors are the
expression levels of 22283 genes. The 71 samples were split with  a roughly 1:1
ratio to form training and testing sets. We report the median of 100
replicates in Table 4. Besides DSDA, the $\ell_1$ logistic
regression \citep{FHT_2008} was also considered because it is an
obvious choice for sparse high-dimensional classification. From
Table 4, it can be seen that both the SSDA methods are significantly more accurate than
DSDA and the $\ell_1$ logistic regression, with SSDA(pooled) yielding the lowest error rate of $1/35$. In addition, the two SSDA methods
select 6 genes, while the other two methods select more than 17 genes.

To gain more insight, we compared the selected genes by SSDA and
those by DSDA or $\ell_1$ logistic regression. In those 100 tries the
2059th gene is most frequently selected by SSDA, but seldom by
DSDA or $\ell_1$ logistic regression. This gene is encoded by $\tt
IRF1$, as it is the first identified interferon regulatory
transcription factor (http://en.wikipedia.org/wiki/IRF1). Discovering the role
of {\tt IRF1} was a major finding in \citet{blood}. Previous studies
show that {\tt IRF1} influences the immune response. Therefore,
healthy and sick people may have different expression levels on this
gene. It is very interesting that we can use a pure statistical
method like SSDA to select {\tt IRF1}. We plot in Figure 1 the
within-group density functions of gene {\tt IRF1} (the 2059th gene).
It can be seen that the raw expression levels of {\tt IRF1} are
skewed, making linear rules unreliable on this gene. After applying
the naive transformation, the distributions of both groups become close to normal, with similar variances. After the pooled transformation, the LDA model becomes even more plausible.

\begin{table}
\caption{Comparison of SSDA(Naive), SSDA(pooled), DSDA and $\ell_1$ logistic
regression on the malaria dataset. The reported numbers are medians
of 100 replicates, with standard errors obtained by bootstrap in
parentheses.}
\begin{center}
\begin{tabular}{|c|cccc|}
\hline
&SSDA&SSDA&DSDA&Logistic\\
&(Naive)&(Pooled)&&\\
\hline
Testing Error&2/35(0.59\%)&1/35(1.35\%)&6/35(0.99\%)&4/35(0.67\%)\\
Fitted Model Size&6(0.4)&6(0.4)&18(1.5)&17(0.6)\\
\hline
\end{tabular}
\end{center}
\end{table}

\begin{figure}
\centering
\makebox{\includegraphics[height=7cm]{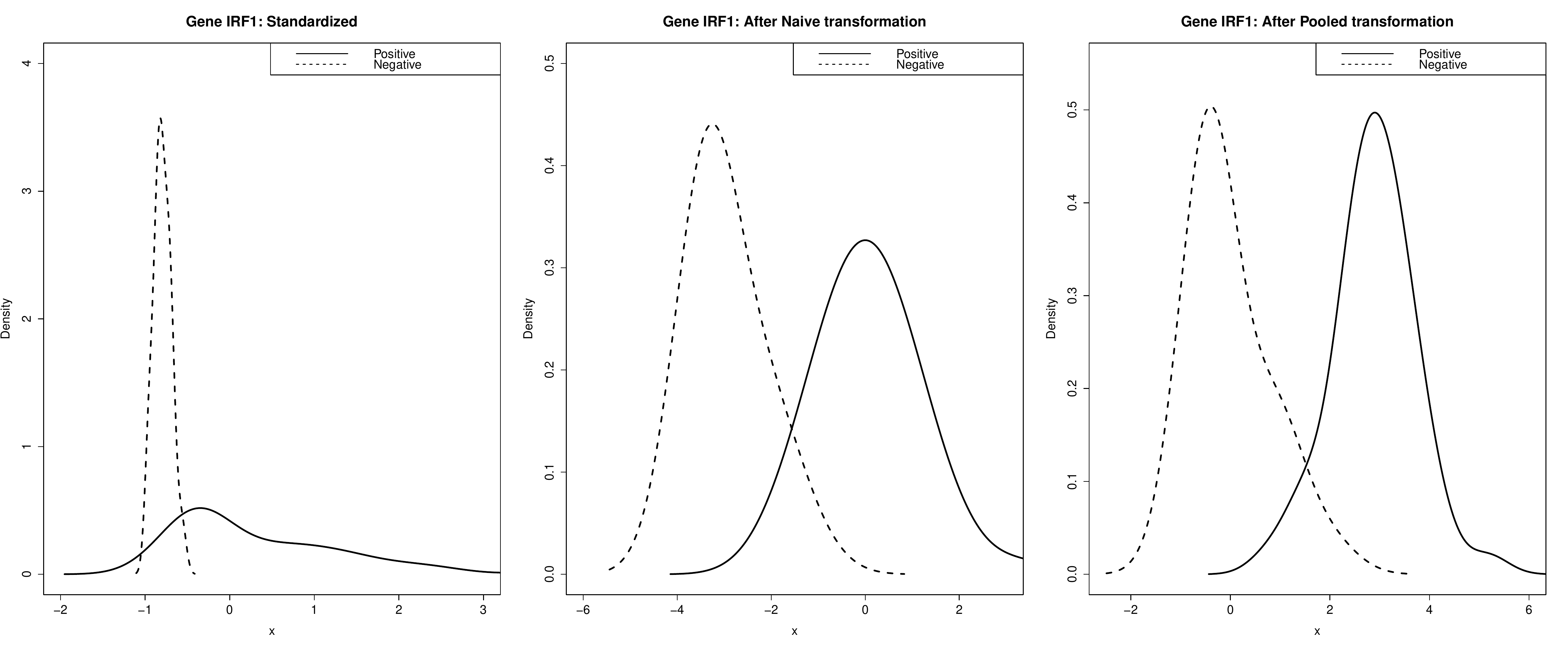}}
\caption{\label{fig01}Density functions of gene {\tt IRF1} (the 2059th gene) in the
malaria data. From left to right the plots display the density functions of
the standardized raw data, naively transformed data as in SSDA(naive) and data after pooled transformation as in SSDA(pooled), respectively.}
\end{figure}

\section{Discussion}
It has been a hot subject of research in recent years to develop
sparse discriminant analysis for high-dimensional classification
and feature selection, rejuvenating the traditional discriminant
analysis. However, sparse discriminant algorithms based on the LDA
model can be very ineffective for non-normal data, as shown in the
simulation study. To overcome the normality limitation, we
consider the semiparametric discriminant analysis model and
propose the SSDA, a high-dimensional semiparametric sparse
discriminant classifier. We have justified SSDA both
theoretically and empirically. For high-dimensional classification
and feature selection, SSDA is more appropriate than the
existing sparse discriminant analysis proposals in the literature.

Although we focus on binary classification throughout the paper, a classifier for multiclass problems is easy to obtain under the semiparametric model. Note that our SSDA method contains two independent steps: transforming the data and fitting a sparse LDA classifier. The first step can be carried out for multiclass problem with proper modification of our pooled estimator, as we will discuss in more detail later, while, in the second step, there already exist multiclass sparse LDA methods, such as sparse optimal scoring \citep{CLT_2008} and $\ell_1$-Fisher's discriminant analysis \citep{WT_2011}. The combination of the transformation and a multiclass sparse LDA method will yield a high-dimensional semiparametric classifier for multiclass problems. Specifically, consider a multiclass model $h(X)\mid Y\sim N(\mu_Y,\Sigma)$ where $Y=1,\ldots, K$ and $\mu_1=0$. Similar to Lemma~\ref{h:twoclass}, we can easily show that $E(\Phi^{-1}\circ F_{kj}(X_j)\mid Y=1)=-\mu_{kj}$ and $E(\Phi^{-1}\circ F_{kj}(X_j)\mid Y=l)=\mu_{lj}-\mu_{kj}$. Define $\hat F_{kj}$ as the empirical CDF of $X_j$ within Class $k$ Winsorized at $(1/n_k^2,1-1/n_k^2)$, where $n_k$ is the sample size within Class $k$. Then we can find
\begin{eqnarray*}
\hat\mu_{kj}^{\mathrm{pool}}&=&\sum_{l=1}^K \hat \pi_{k}\hat \mu_{kj}^{(l)},\,\,\, \hat h_j^{\mathrm{pool}}=\sum_{k=1}^K\hat\pi_k \hat h_{j}^{k} 
\end{eqnarray*}
where $\hat\pi_k=n_k/n$, $\hat \mu_{kj}^{(l)}=\frac{1}{n_k}\sum_{Y^i=k}\Phi^{-1}\circ \hat F_{lj}(X_j^i)-\frac{1}{n_l}\sum_{Y^i=1}\Phi^{-1}\circ \hat F_1(X^i_j)$ and $\hat h_j^{k}=\Phi^{-1}\circ \hat F_{kj}+\hat\mu_{kj}$. With this estimated transformation, one could apply a multiclass sparse LDA method such as the two mentioned above to the pseudo data $(\hat h^{\mathrm{pool}}(X),Y)$.

\section*{Acknowledgement}
Zou’s research is partially supported by NSF grant DMS-08-46068 and the grant N000141110142 from Office of Naval Research. The authors thank the editor, associate editor and two referees for constructive suggestions.

\appendix
\section*{Appendix: proofs}
{
\begin{proof}[Proof of Lemma~\ref{multi-trans}]
First, note that we must have $h_j(X_j)\sim N(0,1)$. Now we show the uniqueness of $h_i$. Suppose $h_i^{(1)}$ and $h_i^{(2)}$ are two strictly increasing transformations such that $h_j^{(1)}(X_j)\sim N(0,1)$,  $h_j^{(2)}(X_j)\sim N(0,1)$. Then for any $t\in \mathbb{R}$, we have
\begin{eqnarray*}
F_j[\{h_j^{(1)}\}^{-1}(t)]&=&\Pr[X_j<\{h_j^{(1)}\}^{-1}(t)]=\Pr\{h_j^{(1)}(X_j)<t\}=\Phi(t)\\
&=&\Pr\{h_j^{(2)}(X_i)<t\}=\Pr[X_j<\{h_j^{(2)}\}^{-1}(t))=F_j[\{h_j^{(2)}\}^{-1}(t)]
\end{eqnarray*}
Because $F_i$ is strictly monotone, we have that $\{h_j^{(1)}\}^{-1}(t)=\{h_j^{(2)}\}^{-1}(t)$ for all $t$, which implies $h_j^{(1)}=h_j^{(2)}$. Now note that $\Phi^{-1}\circ F_j$ is a strictly monotone function that transforms $X_i$ to a standard normal random variable, and the conclusion follows.
\end{proof}
}

\begin{proof}[Proof of Lemma~\ref{h:twoclass}]
By \eqref{SeLDA}, conditional on $Y=-1$, we have $\Phi^{-1}\circ F_{+j}(X_j)\sim N(\mu_{-j},1)$, while, conditional on $Y=+1$, we have $\Phi^{-1}\circ F_{-j}(X_j)+\mu_{-j}\sim N(0,1)$. Hence, the conclusions follow.
\end{proof}

The following properties of the normal distribution are repeatedly used in our proof \citep{SemiCov,ABDJ06}.
\begin{proposition}\label{prop:basic} Let $\phi(t)$ and $\Phi(t)$ be the pdf and CDF of $N(0,1)$, respectively.
\begin{enumerate}
\item For $t \ge 1$,$(2t)^{-1}\phi(t)\le 1-\Phi(t)\le
t^{-1}\phi(t),$
\item For $t
\ge 0.99$, $\Phi^{-1}(t)\le [2\log\{(1-t)^{-1}\}]^{1/2}.$
\end{enumerate}
\end{proposition}


Define
\begin{equation}\label{Adef}
A_n=[-(\gamma_1\log{n})^{1/2}, (\gamma_1\log{n})^{1/2}],
\end{equation}
 where
$0<\gamma_1<1$ is a fixed number and $n$ is the sample size. The
following lemma shows that $\hat h_j(x)$ is an accurate estimator of
$h_j(x)$ for $h_j(x)\in A_n$.
\begin{lemma}\label{A}
For sufficiently large $n$ and $0<\gamma_1<1$, we have
\begin{equation*}
\Pr\{\sup_{h_j(x)\in A_n}|\hat h_j(x)-h_j(x)|\ge \epsilon\}\le
2\exp\{-n^{1-\gamma_1}\epsilon^2/(32\pi^2\gamma_1\log{n})\}+2\exp\{-n^{1-\gamma_1}/(16\pi\gamma_1\log{n})\}.
\end{equation*}
\end{lemma}


\begin{proof}[Proof of Lemma~\ref{A}]
By mean value theorem,
\begin{equation*}
\hat h_j(x)-h_j(x) = (\Phi^{-1})'(\xi)\{\hat F_j(x)-F_j(x)\},
\end{equation*}
for some $\xi \in [\min\{\hat F_j(x),F_j(x)\},\max\{\hat F_j(x),F_j(x)\}]$.

First, we bound $|(\Phi^{-1})'(\xi)|$. This is achieved by bounding
$F_j(x)$ and $\hat F_j(x)$. By definition, for any $h_j(x)\in A_n$,
\begin{eqnarray*}
n ^{-\gamma_1/2}/\{2(2\pi\gamma_1\log{n })^{1/2}\}&\le&
\Phi\{-(\gamma_1\log{n })^{1/2}\}\le F_j(x) \\
&\le&
\Phi\{(\gamma_1\log{n })^{1/2}\}\le 1-n
^{\gamma_1/2}/\{2(2\pi\gamma_1\log{n })^{1/2}\}.
\end{eqnarray*}
On the other hand, for $x$ such that $h_j(x)\in A_n$
\begin{eqnarray*}
&&\Pr[n ^{-\gamma_1/2}/\{4(2\pi\gamma_1\log{n })^{1/2}\}\le \hat F_j(x) \le 1-n ^{-\gamma_1/2}/\{4(2\pi\gamma_1\log{n })^{1/2}\}]\\
&\ge& \Pr[\sup_{h_j(x)\in A_n}|\tilde F_j(x)-F_j(x)|\le n ^{-\gamma_1/2}/\{4(2\pi\gamma_1\log{n })^{1/2}\}]\\
 &\ge&
1-2\exp\{-n^{1-\gamma_1}/(16\pi\gamma_1\log{n})\},
\end{eqnarray*}
where the last inequality follows from Dvoretzky-Kiefer-Wolfowitz
(DKW) inequality.

Consequently, with a probability no less than
$1-2\exp\{-n^{1-\gamma_1}/(16\pi\gamma_1\log{n})\}$,
\begin{equation*}
n ^{-\gamma_1/2}/\{4(2\pi\gamma_1\log{n })^{1/2}\}\le
\xi\le 1-n^{-\gamma_1/2}/\{4(2\pi\gamma_1\log{n })^{1/2}\},
\end{equation*}
 and, combining this fact with Proposition \ref{prop:basic}, we
have
\begin{eqnarray*}
|(\Phi^{-1})'(\xi)|&=&[\phi\{\Phi^{-1}(\xi)\}]^{-1}=(2\pi)^{1/2}\exp\{\Phi^{-1}(\xi)^2/2\}\\
&\le&(2\pi)^{1/2}\exp[\log\{4n^{\gamma_1/2}(2\pi\gamma_1\log{n})^{1/2}\}]\\
&=&8\pi n^{\gamma_1/2}(\gamma_1\log{n})^{1/2}\equiv M_n.
\end{eqnarray*}
 Then
\begin{eqnarray*}
&&\Pr\{\sup_{h_j(x)\in A_n}|\hat h_j(x)-h_j(x)|>\epsilon\}\\
&\le& \Pr\{M_n\sup_{h_j(x) \in A_n}|\hat
F_j(x)-F_j(x)|>\epsilon\}+2\exp\{-n^{1-\gamma_1}/(16\pi\gamma_1\log{n})\}.
\end{eqnarray*}
For the first term on the right hand side,
\begin{eqnarray*}
&&\Pr\{M_n\sup_{h_j(x) \in A_n}|\hat F_j(x)-F_j(x)|>\epsilon\}\\
&\le&\Pr\{M_n\sup_{h_j(x)\in A_n}|\hat F_j(x)-\tilde
F_j(x)|>\epsilon/2\}+\Pr\{M_n\sup_{h_j(x)\in A_n}|
F_j(x)-\tilde F_j(x)|>\epsilon/2\}.\\
\end{eqnarray*}
Because $\sup_{h_j(x)\in A_n}|\hat F_j(x)-\tilde F_j(x)|\le
\delta_n=1/n^2$, $\delta_nM_n\rightarrow 0$ and so the
first term is 0 for sufficiently large $n$. Apply the DKW inequality
to the second term and the conclusion follows.
\end{proof}
The above lemma guarantees that $\hat h_j(X_j)$ is very close to
$h_j(X_j)$ on $A_n$. Now we consider observations in $A_n^c$.
Partition
$A_n^c$ to three regions:
\begin{eqnarray*}
B_n&=&[-\gamma_2\log{n},-(\gamma_1\log{n})^{1/2})\cup
((\gamma_1\log{n})^{1/2},\gamma_2\log{n}];\\
C_n&=&[-n^{\gamma_3},-\gamma_2\log{n})\cup
(\gamma_2\log{n},n^{\gamma_3}];\\
D_n&=&(-\infty,-n^{\gamma_3})\cup (n^{\gamma_3},\infty).
\end{eqnarray*}
Define $\#B_n=\#\{i:h_j(X_j^i)\in B_n\}$ and $\#C_n$, $\#D_n$
analogously.
\begin{lemma}\label{AC}
For sufficiently large $n$ and positive constants
$\alpha_1,\alpha_2$ such that $\alpha_1>1-\gamma_1/2$, we
have
\begin{eqnarray}
\sup_{h_j(x)\in B_n}|\hat h_j(x)-h_j(x)|&\le& 2(\log{n})^{1/2}+\gamma_2\log{n};\label{bound:B}\\
\sup_{h_j(x)\in C_n}|\hat h_j(x)-h_j(x)|&\le&
2(\log{n})^{1/2}+n^{\gamma_3};\label{bound:C}\\
\Pr(\# B_n>n^{\alpha_1})&\le& \exp(-n^{2\alpha_1-1}/4);\label{no:B}\\
\Pr(\# C_n>n^{\alpha_2})&\le& \exp(-n^{2\alpha_2-1}/4);\label{no:C}\\
\Pr(\# D_n >1)&\le&
(2\pi)^{-1/2}2n^{1-\gamma_3}\exp(-n^{2\gamma_3}/2).\label{no:D}
\end{eqnarray}
\end{lemma}

\begin{proof}[Proof of Lemma \ref{AC}]
Equations (\ref{bound:B})--(\ref{bound:C}) are direct consequences
of the definitions of $\hat h$ and $B_n$, $C_n$. Indeed, because
$\hat F<1-\delta_n$, by Proposition \ref{prop:basic}, for $x\in
B_n\cup C_n$
$$
|\hat h_j(x)|\le
\Phi^{-1}(1-\delta_n)\le\{2\log{(\delta_n^{-1})}\}^{1/2}=2(\log{n})^{1/2}.
$$
Combining this bound with the definitions of $B_n$, $C_n$, we have
the desired conclusions.

For (\ref{no:B}), note that, for sufficiently large $n$,
\begin{equation*}
\Pr\{h_j(X_j)\in B_n\}\le 2\Pr\{h_j(X_j)>(\gamma_1\log{n})^{1/2}\}\le
2^{1/2}n^{-\gamma_1/2}/(\pi\gamma_1\log{n})^{1/2}\le n^{-\gamma_1/2}.
\end{equation*}
Therefore, by Hoeffding's inequality
\begin{eqnarray*}
&&\Pr(\# B_n>n^{\alpha_1})\\
&\le& \Pr(\sum_{i=1}^n[I\{h_j(X_j^i)\in B_n\}-\Pr\{h_j(X_j^i)\in
B_n\}]>n^{\alpha_1}-n^{1-\gamma_1/2})\\
&\le&
\exp\{-n^{2\alpha_1-1}(1-n^{1-\gamma_1/2-\alpha_1})^2/2\}\le \exp(-n^{2\alpha_1-1}/4),
\end{eqnarray*}
for sufficiently large $n$.

For (\ref{no:C}), note that \begin{equation*} \Pr\{h_j(X_j^i)\in
C_n\}\le 2n^{-\gamma_2^2\log{n}/2}/\gamma_2\log{n}.
\end{equation*}
 So
(\ref{no:C}) can be proven similarly.

For (\ref{no:D}),
\begin{eqnarray*}
\Pr(\# D_n>1)\le 2n\Pr\{h_j(X^{i}_j)>n^{\gamma_3}\}\le
2n^{1-\gamma_3}(2\pi)^{-1/2}\exp(-n^{2\gamma_3}/2).
\end{eqnarray*}
\end{proof}

\begin{proof}[Proof of Theorem~\ref{accuracy}]
We first prove (\ref{eq:muhat}).

\begin{eqnarray*}
\Pr(|\hat \mu_j-\mu_j|>\epsilon)&\le& \Pr\{n^{-1}\sum_{i=1}^n
|\hat
h_j(X_j^i)-h_j(X_j^i)|>\epsilon/2\}+\Pr\{|n^{-1}\sum_{i=1}^n
h_j(X_j^i)-\mu_j|>\epsilon/2\}\\
&\equiv& L_1+L_2.
\end{eqnarray*}
By the Chernoff bound, $L_2\le 2\exp(-cn\epsilon^2)$.

\begin{eqnarray*}
L_1&\le& \Pr\{\sup_{h_j(x)\in A_n}|\hat
h_j(x)-h_j(x)|>\epsilon/8\}+\Pr\{n^{-1}(\#
B_n)\sup_{h_j(x)\in B_n}|\hat
h_j(x)-h_j(x)|>\epsilon/8\}\\
&+&\Pr\{n^{-1}(\# C_n)\sup_{h_j(x)\in C_n}|\hat
h_j(x)-h_j(x)|>\epsilon/8\}\\
&+&\Pr\{n^{-1}(\#
D_n)\sup_{h_j(x)\in D_n}|\hat
h_j(x)-h_j(x)|>\epsilon/8\}.
\end{eqnarray*}

By Lemma \ref{AC}, it can be checked that, under Condition (C1), if
$\# B_n\le n^{\alpha_1}$ and $\# D_n=0$ then
\begin{eqnarray*}
&&\Pr\{n^{-1}(\#
B_n)\sup_{h_j(x)\in B_n}|\hat
h_j(x)-h_j(x)|>\epsilon/8\}=0,\\
&&\Pr\{n^{-1}(\#
D_n)\sup_{h_j(x)\in D_n}|\hat
h_j(x)-h_j(x)|>\epsilon/8\}=0,
\end{eqnarray*}
 for sufficiently large $n$. If $\gamma_3+\alpha_2<1$,
similarly we have
\begin{equation*}
\Pr\{n^{-1}(\# C_n)\sup_{h_j(x)\in C_n}|\hat
h_j(x)-h_j(x)|>\dfrac{\epsilon}{8}\}=0.
\end{equation*}
 It follows that, if $\alpha_1<1$ and $\gamma_3+\alpha_2<1$, then we
 have
\begin{eqnarray*}
L_1\le
4\exp(-cn^{1-\gamma_1}\epsilon^2/\gamma_1)+\exp(-cn^{2\alpha_1-1})+\exp(-cn^{2\alpha_2-1})+(2\pi)^{-1/2}2n^{1-\gamma_3}\exp(-n^{2\gamma_3}/2),
\end{eqnarray*}
Take $\gamma_1=2\rho, \alpha_1=1-\rho/2,
\alpha_2=3/4-\rho/2,
\gamma_3=1/4-\rho/2$ and the conclusion follows.

Now we prove (\ref{eq:sigmahat}). By the proof in \cite{SemiCov}, it
suffices to bound
\begin{equation*}
\Pr[|n^{-1}\sum_{i=1}^n h_j(X_j^i)\{\hat
h_k(X_k^i)-h_k(X_k^i)\}|>\epsilon].
\end{equation*}
We can decompose the summation into four terms.
\begin{eqnarray*}
&&n^{-1}\sum_{i=1}^n h_j(X_j^i)\{\hat
h_k(X_k^i)-h_k(X_k^i)\}\\
&=&n^{-1}(\sum_{h_j(X_j^i)\in D_n \mbox{ or } h_k(X_k^i)\in
D_n}+\sum_{h_j(X_j^i)\notin D_n, h_k(X_k^i)\in
C_n}\\
&&+\sum_{h_j(X_j^i)\in A_n\cup B_n, h_k(X_k^i)\in
B_n}+\sum_{h_j(X_j^i)\in A_n, h_k(X_k^i)\in A_n})[h_j(X_j^i)\{\hat
h_k(X_k^i)-h_k(X_k^i)\}]\\
&\equiv& S_1+S_2+S_3+S_4.
\end{eqnarray*}
Write $\# D_{nj}=\# \{i: h_j(X_j^i)\in D_n\}$. Then
\begin{eqnarray*}
\Pr(|S_1|>\epsilon)&\le& \Pr(\#D_{nj}>1)+\Pr(\# D_{nk}>1)\\
&\le&
4n^{1-\gamma_3}(2\pi)^{-1/2}\exp(n^{2\gamma_3}/2).\end{eqnarray*}
Note that, for a pair of $\alpha_2, \gamma_3$, such that
$\alpha_2+2\gamma_3-1<0$, we have
$n^{\alpha_2+2\gamma_3-1}\rightarrow 0$. Therefore, for sufficiently
large $n$,
\begin{eqnarray*}
\Pr(|S_2|>\epsilon)&\le& \Pr(n^{-1}\sum_{h_k(X_k^i)\in
C_n}|\hat
h_k(X_k^i)-h_k(X_k^i)|>\epsilon/n^{\gamma_3})\\
&\le& \Pr(\# C_n>n^{\alpha_2})+\Pr[n^{\alpha_2-1}\{2(\log{n})^{1/2}+n^{\gamma_3}\}>\epsilon/n^{\gamma_3}]\\
&\le& \exp(-n^{2\alpha_2-1}/4)+0,
\end{eqnarray*}
Similarly, for $0<\alpha_1<1$,
\begin{eqnarray*}
\Pr(|S_3|>\epsilon)&\le& \Pr(\#
B_n>n^{\alpha_1})+\Pr[n^{\alpha_1-1}(\gamma_2\log{n})\{2(\log{n})^{1/2}+\gamma_2\log{n}\}>\epsilon]\\
&\le&
\exp(-n^{2\alpha_1-1}/4)+0,
\end{eqnarray*}
where $0<\alpha_1<1$. Finally,
\begin{eqnarray*}
\Pr(|S_4|>\epsilon)&\le& \Pr\{\sup_{h_k(X_k^i)\in
A_n}|\hat h_k(X_k^i)-h_k(X_k^i)|>\epsilon(\gamma_1\log{n})^{-1/2}\}\\
&\le& 4\exp\{-cn^{1-\gamma_1}\epsilon^2/(\gamma_1^2\log^2n)\}.
\end{eqnarray*}
Pick $\gamma_1=2\rho,
\gamma_3=1/6-\rho,\alpha_2=2/3-\rho/2,
\alpha_1=1-\rho/2$ and the conclusion follows.
\end{proof}

\begin{proof}[Proof of Corollary~\ref{twoclass}]
Note that $n_+$ is a summation of $n$ i.i.d random variables with
distribution $\mathrm{Bernoulli}(1,\pi_+)$. Therefore, by Chernoff
bound, there exists $c>0$ such that
$\Pr(n_+>\pi_+n/2)>1-2\exp(-cn)$. Hence, by
Theorem~\ref{accuracy},
\begin{equation*}
\Pr(|\hat \mu_{+j}-\mu_{+j}|\ge
\epsilon/2)<\zeta_1^*(\pi_+^{1/2}\epsilon/2)+2\exp(-cn).
\end{equation*}
Similarly,
\begin{equation*}
\Pr(|\hat \mu_{-j}-\mu_{-j}|\ge
\epsilon/2)<\zeta_1^*(\pi_-^{1/2}\epsilon/2)+2\exp(-cn).
\end{equation*}
Hence, we have (\ref{SSeLDA:mu}). Equation (\ref{SeLDA:sigma}) can
be proven similarly.
\end{proof}

\begin{proof}[Proof of Theorem~\ref{thm1} and Theorem~\ref{asy} ]
By \cite{SLDA}, the consistency is implied by accurate estimators of
$\hat \mu_y$, $\hat \sigma_{ij}$. Therefore, Theorem~\ref{thm1} can be
proven by following the proof in their paper and applying
Corollary~\ref{twoclass}.
\end{proof}

Theorem~\ref{asy} is direct consequence of Theorem~\ref{thm1}.
Hence, the proof is omitted here for the sake of space.

\begin{lemma}
For any $\epsilon<\min\{\epsilon_0,\lambda/(2\phi\Delta_1),\lambda\}$ and large enough $n$ such that $\epsilon> sn^{-1/4}$, we have
\begin{enumerate}
\item
\begin{equation}\label{l1beta}
\Pr(\Vert\hat\beta_A-\beta_A\Vert_1\ge \epsilon)\le 2s^2\zeta_2(\epsilon/s)+2s\zeta_1(\epsilon/s).
\end{equation}
\item If we further assume that $\pi_+,\pi_->c>0$, then
\begin{eqnarray}\label{beta0}
\Pr(|\hat\beta_0-\beta_0|\ge c\epsilon)\le 2\exp(-cn)+cs\zeta_1[\epsilon/\{s(\phi\Delta_1+\Delta_2)\}]\\\nonumber
+2p\zeta_1\{\lambda(1-\kappa+2\epsilon\phi)/4(1+\kappa)\}
+2s^2\zeta_2\{\epsilon/(s\Delta_2)\}+2ps\zeta_2(\epsilon/s)
\end{eqnarray}
\end{enumerate}
\end{lemma}
\begin{proof}
We first prove \eqref{l1beta}.
Similar to the proof of Conclusion 3, Theorem 1 in \cite{SLDA}, we have
\begin{equation}
\Vert\hat\beta_A-\beta_A\Vert_1\le (1-\eta_1\phi)^{-1}\{\lambda/2+\phi\Vert(\dhma)-(\dma)\Vert_1+\phi^2\eta_1\Delta_1\}
\end{equation}
where $\eta_1=\Vert \Omega_{AA}-\Omega_{AA}^{(n)}\Vert_{\infty}$. Under the events $\eta_1<\epsilon$ and $\Vert(\dhma)-(\dma)\Vert_1<\epsilon$ we have $\Vert\hat\beta_A-\beta_A\Vert_1\le \epsilon$. Hence, \eqref{l1beta} follows.

For \eqref{beta0}, assume that $\hat\beta_{A^C}=0$. Then we have
\begin{eqnarray*}
|\hat\beta_0-\beta_0|&=&|\{\log{(n_+/n_-)}-\log{(\pi_+/\pi_-)}\}-(\hat\mu_{+A}+\hat\mu_{-A})^\T\hat\beta_A/2+(\mu_{+A}+\mu_{-A})^\T\beta_A/2|\\
&\le&|\log{\hat\pi_+}-\log{\pi_+}|+|\log{\hat\pi_-}-\log{\pi_-}|\\
&&+|\{(\hat\mu_{+A}+\hat\mu_{-A})-(\mu_{+A}+\mu_{-A})\}^\T(\hat\beta_A-\beta_A)|/2\\
&&+|(\mu_{+A}+\mu_{-A})^\T(\hat\beta_A-\beta_A)|+|(\mu_{+A}+\mu_{-A})^\T(\hat\beta_A-\beta_A)|/2\\
\end{eqnarray*}
Under the events $|\hat\pi_j-\pi_j|\le \min\{c/2,2\epsilon/c\}$, $\Vert\hat\mu_{jA}-\mu_{jA}\Vert_1\le\epsilon/\phi\Delta_1$ and $\Vert\hat\beta_A-\beta_A\Vert_1\le\epsilon/\Delta_2$, we have $|\hat\beta_0-\beta_0|\le c\epsilon$.
\end{proof}

\begin{proof}[Proof of Theorem~\ref{error}]
Note that
\begin{eqnarray*}
R_n&\le&1-\Pr(Y=\sign(h(X)^\T\beta+\beta_0), \sign(\hat h(X)^\T\hat\beta+\hat\beta_0)=\sign(h(X)^\T\beta+\beta_0))\\
&\le& R+\Pr(\sign(\hat h(X)^\T\hat\beta+\hat\beta_0)\ne\sign(h(X)^\T\beta+\beta_0))
\end{eqnarray*}
Therefore,
\begin{eqnarray}
R_n-R&\le& \Pr(\sign(\hat h(X)^\T\hat\beta+\hat\beta_0)\ne\sign(h(X)^\T\beta+\beta_0))\\
&\le& \Pr(|h(X)^\T\beta+\beta_0|\le\epsilon)\\\nonumber
&&+\Pr(|(\hat h(X)^\T\hat\beta+\hat\beta_0)-(h(X)^\T\beta+\beta_0)|
\ge\dfrac{\epsilon}{2})
\end{eqnarray}

Now
\begin{equation}
\Pr(|h(X)^\T\beta+\beta_0|\le\epsilon)\le\dfrac{c\epsilon}{\sqrt{2\pi}}
\end{equation}
For the second term, assume that $\hat\beta_{A^C}=0$, $|\hat\beta_0-\beta_0|\le c\epsilon$, $\Vert\hat\beta_A-\beta_A\Vert_1\le \dfrac{\epsilon}{\sqrt{\log{n}}}$ and $\sup_{t\in A_n}|\hat h_j(t)-h_j(t)|\le c\dfrac{\epsilon}{\phi\Delta_1}$ for all $j$, where $A_n$ is defined as in \eqref{Adef}. Then
\begin{eqnarray}
&&|(\hat h(X_A)^\T\hat\beta_A+\hat\beta_0)-(h(X_A)^\T\beta_A+\beta_0)|\\
&\le& |\hat\beta_0-\beta_0|+\Vert\hat h(X_A)\Vert_{\infty}\Vert\hat\beta_A-\beta_A\Vert_1+\Vert\hat h(X_A)- h(X_A)\Vert_{\infty}\Vert\beta_A\Vert_1\\
&\le&|\hat\beta_0-\beta_0|+2\sqrt{\log{n}}\Vert\hat\beta_A-\beta_A\Vert_1+\phi\Delta_1\Vert\hat h(X_A)- h(X_A)\Vert_{\infty},
\end{eqnarray}
which is smaller than $\epsilon$ as long as $h_j(X_j)\in A_n$ for all $j$. Therefore, take $\gamma_1=1/2$ in $A_n$, we have
\begin{equation}
\Pr(|(\hat h(X)^\T\hat\beta+\hat\beta_0)-(h(X)^\T\beta+\beta_0)|
\ge\dfrac{\epsilon}{2})\le \Pr(\cup_{j\in A}h_j(X_j)\in A_n)\le \dfrac{csn^{-1/4}}{\sqrt{\log{n}}},
\end{equation}
which will be smaller than $\epsilon$ for sufficiently large $n$.

Therefore, by Lemma~\ref{A}, \eqref{l1beta}, \eqref{beta0}, we have the desired conclusion.
\end{proof}

\bibliographystyle{agsm}
\bibliography{ref}

\end{document}